\documentclass[10pt,journal,letterpaper,onecolumn]{IEEEtran}

\usepackage{amsfonts,amsthm,amssymb, amsmath, mathtools}
\usepackage[english]{babel}
\usepackage{graphicx,comment} 
\usepackage{epstopdf} 

\newtheorem{theo}{Theorem}[section]

\newtheorem{lemma}[theo]{Lemma}

\theoremstyle{definition} 

\newtheorem{remark}[theo]{Remark}

\numberwithin{equation}{section}

\newcommand{\R}{\mathbb{R}}
\newcommand{\E}{\mathbb{E}}
\newcommand{\N}{\mathbb{N}}
\newcommand{\eps}{\varepsilon}

\renewcommand{\l}{\ell}

\renewcommand{\P}{\mathbb{P}}
\newcommand{\cov}{\mathrm{Cov}}
\newcommand{\var}{\mathrm{Var}}
\newcommand{\id}{\mathrm{Id}}

\newcommand{\A}{\mathcal{A}}

\newcommand{\sgn}{\mathrm{sign}}

\renewcommand{\sp}[2]{\langle #1 , #2 \rangle}

\begin{document}
\title{Non-asymptotic Analysis\\ of $\l_1$-norm Support Vector Machines}
\author{Anton Kolleck\thanks{A. Kolleck is with the Department of Mathematics, Technical University Berlin, Street of 17. June 136, 10623 Berlin, Germany,
(e-mail:{\tt kolleck@math.tu-berlin.de}). A. Kolleck was supported by the DFG Research Center {\sc Matheon} ``Mathematics for key technologies'' in Berlin.},
Jan Vyb\'\i ral\thanks{J. Vyb\'\i ral is with the Department of Mathematical Analysis, Charles University, Sokolovsk\'a 83, 186 00, Prague 8, Czech Republic,
(e-mail: {\tt vybiral@karlin.mff.cuni.cz}). J. Vyb\'\i ral was supported by the ERC CZ grant LL1203 of the Czech Ministry of Education and by
the Neuron Fund for Support of Science.}}
\date{\today}

\maketitle
\begin{abstract}
Support Vector Machines (SVM) with $\ell_1$ penalty 
became a standard tool in analysis
of highdimensional classification problems with sparsity constraints in many applications including bioinformatics and signal processing.
Although SVM have been studied intensively in the literature, this paper has to our knowledge first non-asymptotic results on
the performance of $\ell_1$-SVM in identification of sparse classifiers. We show that a
$d$-dimensional $s$-sparse classification vector can be (with high probability) well
approximated from only $O(s\log(d))$ Gaussian trials.
The methods used in the proof include concentration of measure
and probability in Banach spaces.
\end{abstract}


\section{Introduction}\label{sec:1}

\subsection{Support Vector Machines}

Support vector machines (SVM) are a group of popular classification methods in machine learning. Their input is a set of
data points $x_1,\ldots,x_m\in\R^d$, each equipped with a label $y_i\in\{-1,+1\}$, which assigns each of the
data points to one of two groups. SVM aims for binary linear classification based on separating hyperplane between
the two groups of training data, choosing a hyperplane with separating gap as large as possible.

Since their introduction by Vapnik and Chervonenkis \cite{VaCh}, the subject of SVM was studied intensively.
We will concentrate on the so-called soft margin SVM \cite{CoVa}, which allow also for misclassification of the training data
are the most used version of SVM nowadays.

In its most common form (and neglecting the bias term), the soft-margin SVM is a convex optimization program
\begin{align}
\notag\min\limits_{\substack{w\in\R^d\\\xi\in\R^m}} \frac{1}{2}\|w\|_2^2+\lambda\sum_{i=1}^m \xi_i&\quad\text{subject to}\quad y_i\sp{x_i}{w}\geq 1-\xi_i\\
\label{eq:svm_xi}&\quad\text{and}\quad \xi_i\ge 0
\end{align}
for some tradeoff parameter $\lambda>0$ and so called slack variables $\xi_i$.
It will be more convenient for us to work with the following equivalent reformulation of \eqref{eq:svm_xi}
\begin{align}\label{eq:svm_R}
\min\limits_{w\in\R^d}\sum\limits_{i=1}^{m}[1-y_i\sp{x_i}{w}]_+\quad\text{subject to}\quad\|w\|_2\leq R,
\end{align}
where $R>0$ gives the restriction on the size of $w$.
We refer to monographs \cite{StCh, Vap1, Vap2} and references therein for more details on SVM and to
\cite[Chapter B.5]{FR} and \cite[Chapter 9]{CZ} for a detailed discussion on dual formulations.

\subsection{$\ell_1$-SVM}

As the classical SVM \eqref{eq:svm_xi} and \eqref{eq:svm_R} do not use any pre-knowledge about $w$, one typically
needs to have more training data than the underlying dimension of the problem, i.e. $m\gg d.$
Especially in analysis of high-dimensional data, this is usually not realistic and we typically deal with
much less training data, i.e. with $m\ll d.$ On the other hand, we can often assume some structural assumptions on $w$,
in the most simple case that it is \emph{sparse}, i.e. that most of its coordinates are zero.
Motivated by the success of LASSO \cite{T} in sparse linear regression, it was proposed in \cite{BM'}  that
replacing the $\ell_2$-norm $\|w\|_2$ in \eqref{eq:svm_R} by its $\ell_1$-norm $\|w\|_1=\sum_{j=1}^d |w_j|$
leads to sparse classifiers $w\in\R^d$.
This method was further popularized in \cite{ZRHT} by Zhu, Rosset, Hastie, and Tibshirani, who developed an algorithm
that efficiently computes the whole solution path (i.e. the solutions of \eqref{eq:svm_R} for a wide range of parameters $R>0$).
We refer also to \cite{BM,BM2,M1} and \cite{M2} for other generalizations of the concept of SVM.

Using the ideas of concentration of measure \cite{L} and random constructions in Banach spaces \cite{LT},
the performance of LASSO was analyzed in the recent area of compressed sensing \cite{D,CRT,BCKV,DDEK,FR2}.

$\ell_1$-SVM (and its variants) found numerous applications in high-dimensional data analysis, most notably in bioinformatics
for gene selection and microarray classification \cite{WZZ, ZALP, HK}. Finally, $\ell_1$-SVM's are closely related to other
popular methods of data analysis, like elastic nets \cite{ZH} or sparse principal components analysis \cite{ZHT}.

\subsection{Main results}

The main aim of this paper is to analyze the performance of $\ell_1$-SVM in the non-asymptotic regime.
To be more specific, let us assume that the data points $x_1,\dots,x_m\in\R^d$ can be separated by a hyperplane
according to the given labels $y_1,\dots,y_m\in\{-1,+1\}$, and that this hyperplane is normal to a $s$-sparse vector $a\in\R^d$.
Hence, $\langle a,x_i\rangle > 0$ if $y_i=1$ and $\langle a,x_i\rangle <0$ if $y_i=-1.$
We then obtain $\hat a$ as the minimizer of the $\ell_1$-SVM. The first main result of this paper (Theorem \ref{theo:main_theorem}) then shows
that $\hat a/\|\hat a\|_2$ is a good approximation of $a$, if the data points are i.i.d. Gaussian vectors and the number of measurements
scales linearly in $s$ and logarithmically in $d$. 

Later on, we introduce a modification of $\ell_1$-SVM by adding
an additional $\ell_2$-constraint. It will be shown
in Theorem \ref{theo:l2_main_theorem}, that it still approximates
the sparse classifiers with the number of measurements
$m$ growing linearly in $s$ and logarithmically in $d$,
but the dependence on other parameters improves.
In this sense, this modification outperforms the classical $\ell_1$-SVM.
\subsection{Organization}

The paper is organized as follows.
Section \ref{sec:2} recalls the concept of $\l_1$-Support Vector Machines of \cite{ZRHT}. It includes the main
result, namely Theorem \ref{theo:main_theorem}. It shows
that the $\l_1$-SVM allows to approximate sparse classifier $a$, where the number of measurements only increases
logarithmically in the dimension $d$ as it is typical for several reconstruction algorithms from the field of 
compressed sensing. The two most important ingredients of its proof, Theorems \ref{theo:concentration_estimate}
and \ref{theo:Psi-Psi}, are also discussed in this part. The proof techniques used are based on the recent work of Plan and Vershynin \cite{PV2},
which in turn makes heavy use of classical ideas from the areas of concentration of measure and probability estimates in Banach spaces \cite{L,LT}.

Section \ref{sec:3} gives the proofs of Theorems \ref{theo:concentration_estimate} and \ref{theo:Psi-Psi}. In Section \ref{sec:4}
we discuss several extensions of our work, including a modification of $\ell_1$-SVM, which combines the $\ell_1$ and $\ell_2$ penalty.

Finally, in Section \ref{sec:5} we show numerical tests to demonstrate the convergence results of Section \ref{sec:2}.
In particular, we compare different versions of SVM and 1-Bit Compressed Sensing, 
which was first introduced by Boufounos and Baraniuk in \cite{BB} and then discussed and continued in \cite{PV1,PV2,PVY,AAPV,KSW} and others.

\subsection{Notation}

We denote by $[\lambda]_+:=\max(\lambda,0)$ the positive part
of a real number $\lambda\in\R.$ By $\|w\|_1, \|w\|_2$
and $\|w\|_{\infty}$ we denote the $\ell_1$,
$\ell_2$ and $\ell_\infty$ norm of $w\in\R^d$, respectively.
We denote by ${\mathcal N}(\mu,\sigma^2)$ the normal (Gaussian) distribution
with mean $\mu$ and variance $\sigma^2$. When $\omega_1$ and $\omega_2$ are random variables,
we write $\omega_1\sim \omega_2$ if they are equidistributed.
Multivariate normal distribution is denoted by ${\mathcal N}(\mu,\Sigma)$, where
$\mu\in\R^d$ is its mean and $\Sigma\in\R^{d\times d}$ is its covariance matrix.
By $\log(x)$ we denote the natural logarithm of $x\in(0,\infty)$ with basis $e$.
Further notation will be fixed in Section \ref{sec:2}
under the name of ``Standing assumptions'',
once we fix the setting of our paper.

\section{$\l_1$-norm support vector machines}\label{sec:2}

In this section we give the setting of our study and the main results.
Let us assume that the data points $x_1,\dots,x_m\in\R^d$
are equipped with labels $y_i\in\{-1,+1\}$ in such a way that
the groups $\{x_i:y_i=1\}$ and $\{x_i:y_i=-1\}$
can indeed be separated by a sparse classifier $a$, i.e. that
\begin{align}\label{eq:assumptions-y}
y_i=\sgn(\sp{x_i}{a}),\quad i=1,\ldots,m
\end{align}
and
\begin{align}\label{eq:assumptions-0}
\|a\|_0=\#\{j:a_j\not=0\}\le s.
\end{align}
As the classifier is usually not unique, we cannot identify $a$
exactly by any method whatsoever.
Hence we are interested in a good approximation of $a$
obtained by $\l_1$-norm SVM from a minimal number of training
data. To achieve this goal, we will assume that the training points
\begin{align}\label{eq:assumptions-x}
	x_i=r\tilde x_i,\quad \tilde x_i\sim\mathcal{N}(0,\id)
\end{align}
are i.i.d. measurement vectors for some constant $r>0$.

To allow for more generality, we replace \eqref{eq:assumptions-0} by
\begin{align}\label{eq:assumptions-a}
\|a\|_2=1,\quad \|a\|_1\leq R.
\end{align}
Let us observe, that $\|a\|_2=1$ and $\|a\|_0\le s$ implies also $\|a\|_1\le\sqrt{s}$, i.e. \eqref{eq:assumptions-a} with $R=\sqrt{s}.$

Furthermore, we denote by $\hat a$ the minimizer of 
\begin{align}\label{eq:l1-svm}
	\min\limits_{w\in\R^d}\sum\limits_{i=1}^m[1-y_i\sp{x_i}{w}]_+\quad\text{subject to}\quad\|w\|_1\leq R.
\end{align}

Let us summarize the setting of our work, which we will later on refer to as ``Standing assumptions'' and which we will keep for the rest of this paper.\vskip.2cm

\hskip-.3cm\framebox{\parbox{0.95\columnwidth}{
{\bf Standing assumptions:}
\begin{enumerate}
\item[(i)] $a\in\R^d$ is the true (nearly) sparse classifier with $\|a\|_2=1,\quad \|a\|_1\leq R$, $R\ge 1$, which we want to approximate;
\item[(ii)] $x_i=r\tilde x_i,\quad \tilde x_i\sim\mathcal{N}(0,\id), i=1,\dots,m$ are i.i.d. training data points for some constant $r>0$;
\item[(iii)] $y_i=\sgn(\sp{x_i}{a}),\quad i=1,\ldots,m$ are the labels of the data points;
\item[(iv)] $\hat a$ is the minimizer of \eqref{eq:l1-svm};
\item[(v)] Furthermore, we denote
\begin{align}
	K&=\{w\in\R^d\mid \|w\|_1\leq R\},\label{eq:K}\\
	f_a(w)&=\frac{1}{m}\sum\limits_{i=1}^m[1-y_i\sp{x_i}{w}]_+,\label{eq:f_a}
\end{align}
where the subindex $a$ denotes the dependency of $f_a$ on $a$ (via $y_i$).
\end{enumerate}
}}\vskip.2cm

In order to estimate the difference between $a$ and $\hat a$
we adapt the ideas of \cite{PV2}. First we observe
\begin{align*}
0&\leq f_a(a)-f_a(\hat a)\\
&=\big(\E f_a(a)-\E f_a(\hat a)\big) + \big(f_a(a)-\E f_a(a)\big) \\
&\qquad-\big(f_a(\hat a)-\E f_a(\hat a)\big)\\
&\leq \E(f_a(a)-f_a(\hat a)) + 2\sup\limits_{w\in K}\vert f_a(w)-\E f_a(w)\vert,
\end{align*}
i.e.
\begin{align}\label{eq:main_idea}
	\E(f_a(\hat a)-f_a(a))
	\leq 2\sup\limits_{w\in K}\vert f_a(w)-\E f_a(w)\vert.
\end{align}
Hence, it remains 
\begin{itemize}
\item[$\bullet$] to bound the right hand side of \eqref{eq:main_idea} from above and
\item[$\bullet$] to estimate the left hand side in \eqref{eq:main_idea} by the distance between $a$ and
$\hat a$ from below.
\end{itemize}

We obtain the following two theorems, whose proofs are given in Section \ref{sec:3}.
	
\begin{theo}\label{theo:concentration_estimate}
Let $u>0$. Under the ``Standing assumptions'' it holds
	\begin{align*}
		\sup\limits_{w\in K}\vert f_a(w)-\E f_a(w)\vert
		\leq \frac{8\sqrt{8\pi}+18rR\sqrt{2\log(2d)}}{\sqrt{m}}+u
	\end{align*}
	with probability at least
	\begin{align*}
		1-8\biggl(\exp\biggl(\frac{-mu^2}{32}\biggr)+\exp\biggl(\frac{-mu^2}{32r^2R^2}\biggr)\biggr).
	\end{align*}
\end{theo}
	
\begin{theo}\label{theo:Psi-Psi}
Let the ``Standing assumptions'' be fulfilled and let $w\in K$. Put
	\begin{align*}
		c=\sp{a}{w},\quad c'=\sqrt{\|w\|_2^2-\sp{a}{w}^2}
	\end{align*}
and assume that $c'>0$. If furthermore $c\leq 0$, then $\pi\E(f_a(w)-f_a(a))$ can be estimated from below by
	\begin{align*}
		\frac{\pi}{2}+c'r\frac{\sqrt{\pi}}{\sqrt{2}}-\frac{\sqrt{2\pi}}{r}.
\end{align*}
If $c>0$, then $\pi\E(f_a(w)-f_a(a))$ can be estimated from below by
\begin{align*}
\frac{\sqrt{\pi}}{\sqrt{2}}\int_0^{1/cr}(1-crt)e^{\frac{-t^2}{2}}\,dt
+\frac{c'}{c}\exp\left(\frac{-1}{2c^2r^2}\right)-\frac{\sqrt{2\pi}}{r}.
\end{align*}
\end{theo}

Combining Theorems \ref{theo:concentration_estimate} and \ref{theo:Psi-Psi} with \eqref{eq:main_idea} we obtain our main result.

\begin{theo}\label{theo:main_theorem}
Let $d\ge 2$, $0<\eps<0.18$, $r>\sqrt{2\pi}(0.57-\pi\eps)^{-1}$ and $m\geq C\eps^{-2}r^2R^2\log(d)$ for some constant $C$.
Under the ``Standing assumptions'' it holds
\begin{align}\label{eq:main_theorem}
	\frac{\left\|a-\frac{\hat a}{\|\hat a\|_2}\right\|_2}{\sp{a}{\frac{\hat a}{\|\hat a\|_2}}}
	\leq C'\left(\eps+\frac{1}{r}\right)
\end{align}
with probability at least
\begin{align}\label{eq:prob:1}
1-\gamma\exp\left(-C''\log(d)\right)
\end{align}
for some positive constants $\gamma,C',C''$.
\end{theo}
\begin{remark}
\begin{enumerate}
\item If the classifier $a\in\R^d$ with $\|a\|_2=1$ is $s$-sparse, we always have $\|a\|_1\le \sqrt{s}$
and we can choose $R=\sqrt{s}$ in Theorem \ref{theo:main_theorem}. The dependence of $m$, the number of samples needed,
is then linear in $s$ and logarithmic in $d$. Intuitively, this is the best what we can hope for. On the other hand,
we leave it open, if the dependence on $\eps$ and $r$ is optimal in Theorem \ref{theo:main_theorem}.
\item Theorem \ref{theo:main_theorem} uses the constants $C$, $C'$ and $C''$ only for simplicity. More explicitly we show that taking
		\begin{align*}
			m\geq 4\eps^{-2}\left(8\sqrt{8\pi}+19rR\sqrt{2\log(2d)}\right)^2,
		\end{align*}
		we get the estimate
		\begin{align*}
			\frac{\|a-\hat a/\|\hat a\|_2\|_2}{\sp{a}{\hat a/\|\hat a\|_2}}
			\leq 2 e^{1/2}\left(\pi\eps+\frac{\sqrt{2\pi}}{r}\right)
		\end{align*}
		with probability at least
		\begin{align*}
			1-8\biggl(\exp\biggl(\frac{-r^2R^2\log(2d)}{16}\biggr)+\exp\biggl(\frac{-\log(2d)}{16}\biggr)\biggr).
		\end{align*}
\item If we introduce an additional parameter $t>0$ and choose 
$m\geq 4\eps^{-2}(8\sqrt{8\pi}+(18+t)rR\sqrt{2\log(2d)})^2$, nothing but the probability changes to 
\begin{align*}
1-8\biggl(\exp\biggl(\frac{-t^2r^2R^2\log(2d)}{16}\biggr)+\exp\biggl(\frac{-t^2\log(2d)}{16}\biggr)\biggr).
\end{align*}
Hence, by fixing $t$ large, we can increase the value of $C''$ and speed up the convergence of \eqref{eq:prob:1} to 1.
\end{enumerate}
\end{remark}
\begin{proof}[Proof of Theorem \ref{theo:main_theorem}]
	To apply Theorem \ref{theo:concentration_estimate} we choose 
	\begin{align*}
		u&=\frac{rR\sqrt{2\log(2d)}}{\sqrt{m}}\quad \intertext{and}\quad 
		m&\geq 4\eps^{-2}(8\sqrt{8\pi}+19rR\sqrt{2\log(2d)})^2
	\end{align*}
	and we obtain the estimate
	\begin{align*}
		\sup\limits_{w\in K}\vert f_a(w)-\E f_a(w)\vert
		\leq\frac{8\sqrt{8\pi}+18rR\sqrt{2\log(2d)}}{\sqrt{m}}+u
		\leq\frac{\eps}{2}
	\end{align*}
	with probability at least
	\begin{align*}
		&\phantom{=.}1-8\biggl(\exp\biggl(\frac{-mu^2}{32}\biggr)+\exp\biggl(\frac{-mu^2}{32r^2R^2}\biggr)\biggr)\\
		&=1-8\biggl(\exp\biggl(\frac{-r^2R^2\log(2d)}{16}\biggr)+\exp\biggl(\frac{-\log(2d)}{16}\biggr)\biggr).
	\end{align*}
Using \eqref{eq:main_idea} this already implies
\begin{align}\label{eq:right_inequality}
	\E\big(f_a(\hat a)-f_a(a)\big)\leq\eps
\end{align}
with at least the same probability. Now we want to apply Theorem \ref{theo:Psi-Psi} with $w=\hat a$ to estimate the left hand side of this inequality.
Therefore we first have to deal with the case $c'=\sqrt{\|\hat a\|_2^2-\sp{a}{\hat a}^2}=0$, 
which only holds if $\hat a=\lambda a$ for some $\lambda\in\R$. If $\lambda>0$, then $\hat a/\|\hat a\|_2=a$ and the statement of the Theorem
holds trivially. If $\lambda\leq 0$, then the condition $f(\hat a)\le f(a)$ can be rewritten as
\begin{align*}
\sum_{i=1}^m[1+|\lambda|\cdot|\langle x_i,a\rangle|]_+\le \sum_{i=1}^m[1-|\langle x_i,a \rangle|]_+.
\end{align*}
This inequality holds if, and only if, $\langle x_i,a\rangle=0$ for all $i=1,\dots,m$ - and this in turn happens only with probability zero.

We may therefore assume that $c'\neq0$ holds almost surely and we can apply Theorem \ref{theo:Psi-Psi}.
	Here we distinguish the three cases $c=\sp{\hat a}{a}\leq0$, $0<c\leq 1/r$ and $1/r<c$. First, we will show that the 
	two cases $c\leq0$ and $0<c<1/r$ lead to a contradiction and then, for the case $c>1/r$, we will prove our claim.\\
\emph{1. case $c\le 0$:} Using Theorem \ref{theo:Psi-Psi} we get the estimate
	\begin{align*}
		\pi\E(f_a(\hat a)-f_a(a))
		\geq\frac{\pi}{2}+c'r\frac{\sqrt{\pi}}{\sqrt{2}}-\frac{\sqrt{2\pi}}{r}
		\geq\frac{\pi}{2}-\frac{\sqrt{2\pi}}{r}
	\end{align*}
	and \eqref{eq:right_inequality} gives (with our choices for $r$ and $\eps$) the contradiction
	\begin{align*}
		\frac{1}{\pi}\left(\frac{\pi}{2}-\frac{\sqrt{2\pi}}{r}\right)
		\leq\E(f_a(\hat a)-f_a(a))
		\leq\eps.
	\end{align*}
\emph{2. case $0<c\leq 1/r$:} 
	As in the first case we use Theorem \ref{theo:Psi-Psi} in order to show a contradiction. First we get the estimate
	\begin{align*}
		\pi&\E(f_a(\hat a)-f_a(a))\\
		&\geq \frac{\sqrt{\pi}}{\sqrt{2}}\int_0^{1/cr}(1-crt)e^{\frac{-t^2}{2}}\,dt
			+\frac{c'}{c}\exp\left(\frac{-1}{2c^2r^2}\right)-\frac{\sqrt{2\pi}}{r}\\
		&\geq \frac{\sqrt{\pi}}{\sqrt{2}}\int_0^{1/cr}(1-crt)e^{\frac{-t^2}{2}}\,dt-\frac{\sqrt{2\pi}}{r}.
	\end{align*}
	Now we consider the function
	\begin{align*}
		g\colon(0,\infty)\to\R,\quad z\mapsto \int_0^{1/z}(1-zt)e^{\frac{-t^2}{2}}\,dt.
	\end{align*}
	It holds $g(z)\geq 0$ and
	\begin{align*}
		g'(z)=-\int_0^{1/z}te^{\frac{-t^2}{2}}\,dt<0,
	\end{align*}
	so $g$ is monotonic decreasing. With $cr<1$ this yields
	\begin{align*}
		\pi\E(f_a(\hat a)&-f_a(a))
		\geq \frac{\sqrt{\pi}}{\sqrt{2}}\int_0^{1/cr}(1-crt)e^{\frac{-t^2}{2}}\,dt-\frac{\sqrt{2\pi}}{r}\\
		&= \frac{\sqrt{\pi}}{\sqrt{2}}g(cr)-\frac{\sqrt{2\pi}}{r}
		\geq \frac{\sqrt{\pi}}{\sqrt{2}}g(1)-\frac{\sqrt{2\pi}}{r}\\
		&=\frac{\sqrt{\pi}}{\sqrt{2}}\int_0^1(1-t)e^{\frac{-t^2}{2}}\,dt-\frac{\sqrt{2\pi}}{r}\\
		&\geq 0.57-\frac{\sqrt{2\pi}}{r}.
	\end{align*}
	Again, \eqref{eq:right_inequality} now gives the contradiction
	\begin{align*}
		\frac{1}{\pi}\left(0.57-\frac{\sqrt{2\pi}}{r}\right)
		\leq\E(f_a(\hat a)-f_a(a))\leq\eps.
	\end{align*}
	We conclude that it must hold $c'>0$ and $c>1/r$ almost surely.

\emph{3. case $1/r<c$:}
	In this case we get the estimate
	\begin{align}\notag
		\pi\E(f_a(\hat a)-f_a(a))
		&\geq \frac{\sqrt{\pi}}{\sqrt{2}}\int_0^{1/cr}(1-crt)e^{\frac{-t^2}{2}}\,dt\\
		&\qquad+\frac{c'}{c}\exp\left(\frac{-1}{2c^2r^2}\right)-\frac{\sqrt{2\pi}}{r}\label{eq:first_estimate_left_side}
\\
		&\notag\geq \frac{c'}{c}\exp\left(\frac{-1}{2c^2r^2}\right)-\frac{\sqrt{2\pi}}{r}\\
		&\notag\geq \frac{c'}{c}e^{-1/2}-\frac{\sqrt{2\pi}}{r},
	\end{align}
	where we used $cr>1$ for the last inequality. Further we get
	\begin{align}\notag
		\frac{c'}{c}
		&=\frac{\sqrt{\|\hat a\|_2^2-\sp{a}{\hat a}^2}}{\sp{a}{\hat a}}
		=\sqrt{\frac{\|\hat a\|_2^2-\sp{a}{\hat a}^2}{\sp{a}{\hat a}^2}}\\
		&=\sqrt{\left(\frac{\|\hat a\|_2-\sp{a}{\hat a}}{\sp{a}{\hat a}}\right)
			\left(\frac{\|\hat a\|_2+\sp{a}{\hat a}}{\sp{a}{\hat a}}\right)}\notag\\
		&=\sqrt{\frac{(2-2\sp{a}{\hat a/\|\hat a\|_2})(2+2\sp{a}{\hat a/\|\hat a\|_2})}{4\sp{a}{\hat a/\|\hat a\|_2}^2}}\label{eq:second_estimate_left_side}
\\
		&\notag=\sqrt{\frac{\|a-\hat a/\|\hat a\|_2\|_2^2\cdot\|a+\hat a/\|\hat a\|_2\|_2^2}{4\sp{a}{\hat a/\|\hat a\|_2}^2}}\\
		&\notag\geq\frac{1}{2}\frac{\|a-\hat a/\|\hat a\|_2\|_2}{\sp{a}{\hat a/\|\hat a\|_2}}.
	\end{align}
	Finally, combining \eqref{eq:right_inequality}, \eqref{eq:first_estimate_left_side} and
	\eqref{eq:second_estimate_left_side}, we arrive at
	\begin{align*}
		\frac{1}{\pi}&\left(\frac{\|a-\hat a/\|\hat a\|_2\|_2}{\sp{a}{\hat a/\|\hat a\|_2}}
			\frac{1}{2}e^{-1/2}-\frac{\sqrt{2\pi}}{r}\right)\\
		&\qquad\leq\E(f_a(\hat a)-f_a(a))
		\leq\eps,
	\end{align*}
	which finishes the proof of the theorem.
\end{proof}

\section{Proofs}\label{sec:3}

The main aim of this section is to prove Theorems \ref{theo:concentration_estimate} and \ref{theo:Psi-Psi}.
Before we come to that, we shall give a number of helpful Lemmas.

\subsection{Concentration of $f_a(w)$}
In this subsection we want to show that $f_a(w)$ does not deviate uniformly far from its expected value $\E f_a(w)$,
i.e. we want to show that 
\begin{align*}
	\sup\limits_{w\in K}\vert f_a(w)-\E f_a(w)\vert
\end{align*}
is small with high probability. Therefore we will first estimate its mean 
\begin{align}
	\mu:=\E\biggl(\sup\limits_{w\in K}\vert f_a(w)-\E f_a(w)\vert\biggr)\label{eq:mu}
\end{align}
and then use a concentration inequality to prove Theorem \ref{theo:concentration_estimate}.
The proof relies on standard techniques from \cite{LT} and \cite{L}
and is inspired by the analysis of 1-bit compressed sensing given in \cite{PV2}.

For $i=1,\ldots,m$ let $\eps_i\in\{+1,-1\}$ be i.i.d. Bernoulli variables with
\begin{align}\label{eq:bernoulli}
\P(\eps_i=1)=\P(\eps_i=-1)=1/2.
\end{align}
Let us put
\begin{align}\label{eq:Ai_and_Ai'}
	\A_i(w)=[1-y_i\sp{x_i}{w}]_+,\quad
	\A(w)=[1-y\sp{x}{w}]_+,
\end{align}
where $x$ is an independent copy of any of the $x_i$ and $y=\sgn(\langle x,a\rangle)$.
Further, we will make use of the following lemmas.

\begin{lemma}\label{lemma:bernoulli}
	For $m\in\N$, i.i.d. Bernoulli variables $\eps_1,\ldots,\eps_m$ according to \eqref{eq:bernoulli} and any scalars
	$\lambda_1,\ldots,\lambda_m\in\R$ it holds 
	\begin{align}\label{eq:bernoulli_sum}
		\P\biggl(\sum\limits_{i=1}^m\eps_i[\lambda_i]_+\geq t\biggr)
		\leq 2\P\biggl(\sum\limits_{i=1}^m\eps_i \lambda_i\geq t\biggr).
	\end{align}
\end{lemma}
\begin{proof}
	First we observe
	\begin{align*}
		&\P\biggl(\sum\limits_{i=1}^m\eps_i[\lambda_i]_+\geq t\biggr)
		=\P\biggl(\sum\limits_{\lambda_i\geq0}\eps_i\lambda_i\geq t\biggr)\\
		&\qquad=\P\biggl(\sum\limits_{\lambda_i\geq0}\eps_i\lambda_i\geq t\text{ and }
			\sum\limits_{\lambda_i<0}\eps_i\lambda_i\geq 0\biggr)\\
			&\quad\qquad+\P\biggl(\sum\limits_{\lambda_i\geq0}\eps_i\lambda_i\geq t\text{ and }
			\sum\limits_{\lambda_i<0}\eps_i\lambda_i< 0\biggr).
	\end{align*}
	Now we can estimate the second of these two probabilities by the first one and we arrive at
	\begin{align*}
		\P\biggl(\sum\limits_{i=1}^m\eps_i[\lambda_i]_+\geq t\biggr)
		&\leq 2\P\biggl(\sum\limits_{\lambda_i\geq0}\eps_i\lambda_i\geq t\text{ and }
			\sum\limits_{\lambda_i<0}\eps_i\lambda_i\geq 0\biggr)\\
		&\leq 2\P\biggl(\sum\limits_{i=1}^m\eps_i\lambda_i\geq t\biggr).
	\end{align*}
\end{proof}

\begin{lemma}\label{lemma:concentration_bernoulli_gaussian}
	\begin{enumerate}
	\item For Gaussian random variables $x_1,\ldots,x_m\in\R^d$ according to \eqref{eq:assumptions-x} it holds
		\begin{align}\label{eq:bernoulli_expectation}
			\E\biggl\|\frac{1}{m}\sum\limits_{i=1}^m x_i\biggr\|_\infty
			\leq\frac{r\sqrt{2\log(2d)}}{\sqrt{m}}.
		\end{align}
	\item Let the i.i.d. Bernoulli variables $\eps_1,\ldots,\eps_m$ be according to \eqref{eq:bernoulli} and let $u>0$. Then it holds
		\begin{align}\label{eq:bernoulli_concentration}
			\P\biggl(\biggl\vert\frac{1}{m}\sum\limits_{i=1}^m\eps_i\biggr\vert\geq u\biggr)
			\leq 2\exp\biggl(\frac{-mu^2}{2}\biggr).
		\end{align}
	\item For $x_1,\ldots,x_m\in\R^d$ and $K\subset\R^d$ according to \eqref{eq:assumptions-x} and \eqref{eq:K} we denote
	\begin{align}\label{eq:tildemu}
	\tilde\mu=\E\biggl(\sup\limits_{w\in K}\biggl\langle\frac{1}{m}\sum\limits_{i=1}^m x_i,w\biggr\rangle\biggr).
	\end{align}
Then it holds
		\begin{align}\label{eq:gaussian_concentration}
			\P\biggl(\sup\limits_{w\in K}\biggl\langle\frac{1}{m}\sum\limits_{i=1}^m x_i,w\biggr\rangle
			\geq \tilde\mu+u\biggr)
			\leq \exp\biggl(\frac{-mu^2}{2r^2R^2}\biggr).
		\end{align}
	\end{enumerate}
\end{lemma}
\begin{proof}
	\begin{enumerate}
	\item The statement follows from
		\begin{align*}
			\E\biggl\|\frac{1}{m}\sum\limits_{i=1}^mx_i\biggr\|_\infty
			=\frac{r}{\sqrt{m}}\E\|\tilde x\|_\infty
		\end{align*}
		with $\tilde x\sim\mathcal{N}(0,\id)$ and proposition 8.1 of \cite{FR}:
		\begin{align}\label{eq:gaussian_uniform_expectation}
			\frac{\sqrt{\log(d)}}{4}\leq\E\|\tilde x\|_\infty\leq\sqrt{2\log(2d)}.
		\end{align}
	\item The estimate follows as a consequence of Hoeffding's inequality \cite{Hoef}. 
	\item Theorem 5.2 of \cite{PV2} gives the estimate
		\begin{align*}
			&\phantom{=.}\P\biggl(\sup\limits_{w\in K}\biggl\langle\frac{1}{m}\sum\limits_{i=1}^m x_i,w\biggr\rangle
			\geq 
			\tilde \mu
				+u\biggr)
			\leq\exp\biggl(\frac{-u^2}{2\sigma^2}\biggr)
		\end{align*}
		with
		\begin{align*}
			\sigma^2
			=\sup\limits_{w\in K}\E\biggl(\biggl\langle\frac{1}{m}\sum\limits_{i=1}^m x_i,w\biggr\rangle^2\biggr).
		\end{align*}
		Since the $x_i's$ are independent we get
		\begin{align*}
			\frac{1}{m}\sum\limits_{i=1}^m x_i
			=\frac{r}{\sqrt{m}}\tilde x\quad\text{with}\quad\tilde x\sim\mathcal{N}(0,\id)
		\end{align*}
		and we end up with
		\begin{align}\label{eq:mot:12}
			\sigma^2
			=\sup\limits_{w\in K}\E\biggl(\frac{r^2}{m}\langle \tilde x,w\rangle^2\biggr)
			=\frac{r^2}{m}\sup\limits_{w\in K}\|w\|_2^2
			=\frac{r^2R^2}{m}.
		\end{align}
	\end{enumerate}
\end{proof}

\subsubsection{Estimate of the mean $\mu$}

To estimate the mean $\mu$, we first derive the following symmetrization inequality, cf. \cite[Chapter 6]{LT} and \cite[Lemma 5.1]{PV2}.

\begin{lemma}[Symmetrization]\label{lemma:symmetrization}
Let $\eps_1,\ldots,\eps_m$ be i.i.d. Bernoulli variables according to \eqref{eq:bernoulli}. Under the ``Standing assumptions''
it holds for $\mu$ defined by \eqref{eq:mu}
	\begin{align}\label{eq:estimate_mu}
		\mu
		\leq 2\E\sup\limits_{w\in K}\biggl\vert\frac{1}{m}\sum\limits_{i=1}^m \eps_i[1-y_i\sp{x_i}{w}]_+\biggr\vert.
	\end{align}
\end{lemma}
\begin{proof}
Let $\A_i(w)$ and $\A(w)$ be according to \eqref{eq:Ai_and_Ai'}. Let $x_i'$ and $x'$ be independent copies of $x_i$ and $x$.
Then $\A_i'(w)$ and $\A'(w)$, generated in the same way \eqref{eq:Ai_and_Ai'} with $x_i'$ and $x'$ instead of $x_i$ and $x$,
are independent copies of $\A_i(w)$ and $\A(w)$. We denote by $\E'$ the mean value with respect
to $x_i'$ and $x'$.
Using $\E'\big(\A'_i(w)-\E'\A'(w)\big)=0$, we get
\begin{align*}
\mu
&=\E\sup\limits_{w\in K}\biggl\vert\frac{1}{m}\sum\limits_{i=1}^m\big(\A_i(w)
-\E\A(w)\big)\biggr|\\
&=\E\sup\limits_{w\in K}\bigg\vert\frac{1}{m}\sum\limits_{i=1}^m\big(\A_i(w)-\E\A(w)\big)\\
&\qquad\qquad-\E'\big(\A'_i(w)-\E'\A'(w)\big)\bigg\vert\\
&=\E\sup\limits_{w\in K}\biggl\vert\frac{1}{m}\sum\limits_{i=1}^m\E'\big(\A_i(w)-\A'_i(w)\big)\biggr\vert.
\end{align*}
Applying Jensen's inequality 
we further get
\begin{align*}
\mu&\le\E\,\E'\sup\limits_{w\in K}\biggl\vert\frac{1}{m}\sum\limits_{i=1}^m\big(\A_i(w)-\A'_i(w)\big)\biggr\vert\\
&=\E\,\E'\sup\limits_{w\in K}\biggl\vert\frac{1}{m}\sum\limits_{i=1}^m\eps_i\big(\A_i(w)
	-\A_i'(w)\big)\biggr\vert\\
&\leq 2\E\sup\limits_{w\in K}\biggl\vert\frac{1}{m}\sum\limits_{i=1}^m\eps_i\A_i(w)\biggr\vert\\
&=2\E\sup\limits_{w\in K}\biggl\vert\frac{1}{m}\sum\limits_{i=1}^m \eps_i[1-y_i\sp{x_i}{w}]_+\biggr\vert
\end{align*}
	as claimed.
\end{proof}

Equipped with this tool, we deduce the following estimate for $\mu$.
	
\begin{lemma}\label{lemma:estimate_mu}
Under the ``Standing assumptions'' we have
	\begin{align*}
		\mu=\E\sup\limits_{w\in K}\vert f_a(w)-\E f_a(w)\vert\leq \frac{4\sqrt{8\pi}+8rR\sqrt{2\log(2d)}}{\sqrt{m}}.
	\end{align*}
\end{lemma}
\begin{proof}
	Using Lemma \ref{lemma:symmetrization} we obtain 
	\begin{align*}
		\mu
		&=\E\sup\limits_{w\in K}\vert f_a(w)-\E f_a(w)\vert\\
		&\leq 2\E\sup\limits_{w\in K}\biggl\vert\frac{1}{m}\sum\limits_{i=1}^m\eps_i[1-y_i\sp{x_i}{w}]_+\biggr\vert\\
		&=2\int_0^\infty\P\biggl(\sup\limits_{w\in K}\biggl\vert
			\frac{1}{m}\sum\limits_{i=1}^m\eps_i[1-y_i\sp{x_i}{w}]_+\biggr\vert \ge t\biggr)\,dt.
	\end{align*}
	Now we can apply Lemma \ref{lemma:bernoulli} to get
	\begin{align*}
		\mu&\leq4\int_0^\infty\P\biggl(\sup\limits_{w\in K}\biggl\vert
			\frac{1}{m}\sum\limits_{i=1}^m\eps_i(1-y_i\sp{x_i}{w})\biggr\vert \ge t\biggr)\,dt\\
		&\leq4\int_0^\infty\P\biggl(\biggl\vert\frac{1}{m}\sum\limits_{i=1}^m\eps_i\biggr\vert\ge t/2\biggr)\\
			&\qquad +\P\biggl(\sup\limits_{w\in K}\biggl\vert
			\frac{1}{m}\sum\limits_{i=1}^m\eps_iy_i\sp{x_i}{w}\biggr\vert \ge t/2\biggr)\,dt.
	\end{align*}	
	Using the second part of Lemma \ref{lemma:concentration_bernoulli_gaussian} we can further estimate
	\begin{align*}
		\mu&\leq \frac{4\sqrt{8\pi}}{\sqrt{m}}+4\int_0^\infty\P\biggl(\sup\limits_{w\in K}\biggl\vert
			\frac{1}{m}\sum\limits_{i=1}^m\eps_iy_i\sp{x_i}{w}\biggr\vert \ge t/2\biggr)\,dt\\
		&=\frac{4\sqrt{8\pi}}{\sqrt{m}}+8\E\biggl(\sup\limits_{w\in K}\biggl\vert
			\biggl\langle\frac{1}{m}\sum\limits_{i=1}^m\eps_i x_i,w\biggr\rangle\biggr\vert\biggr).
	\end{align*}
Using the duality $\|\cdot\|_1'=\|\cdot\|_\infty$ and 
the first part of Lemma \ref{lemma:concentration_bernoulli_gaussian} we get
\begin{align*}
	=\frac{4\sqrt{8\pi}}{\sqrt{m}}+8R\,\E\biggl\|\frac{1}{m}\sum\limits_{i=1}^mx_i\biggr\|_\infty
	\leq \frac{4\sqrt{8\pi}}{\sqrt{m}}+\frac{8rR\sqrt{2\log(2d)}}{\sqrt{m}}.
\end{align*}
\end{proof}

\subsubsection{Concentration inequalities}

In this subsection we will estimate the probability that $f_a(w)$
deviates anywhere on $K$ far from its mean,
i.e. the probability
\begin{align*}
	\P\biggl(\sup\limits_{w\in K}\vert f_a(w)-\E f_a(w)\vert\geq \mu+t\biggr)
\end{align*}
for some $t>0$. First we obtain the following modified version of the second part of
Lemma 5.1 of \cite{PV2}, cf. also \cite[Chapter 6.1]{LT}.

\begin{lemma}[Deviation inequality]\label{lemma:deviation_inequality}
Let $\eps_1,\ldots,\eps_m$ be i.i.d. Bernoulli variables according to \eqref{eq:bernoulli} and let the ``Standing assumptions'' be fulfilled.
Then, for $\mu\in\R$ according to \eqref{eq:mu} and any $t>0$, it holds
	\begin{align}\label{eq:deviation_inequality}
		\P&\biggl(\sup\limits_{w\in K}\vert f_a(w)-\E f_a(w)\vert\geq 2\mu+t\biggr)\\
		&\notag\qquad\leq 4\P\biggl(\sup\limits_{w\in K}\biggl\vert\frac{1}{m}\sum\limits_{i=1}^m\eps_i[1-y_i\sp{x_i}{w}]_+\biggr\vert
			\geq t/2\biggr).
	\end{align}
\end{lemma}
\begin{proof}
	Using Markov's inequality 
let us first note
	\begin{align*}
		\P&\biggl(\sup\limits_{w\in K}\vert f_a(w)-\E f_a(w)\vert\geq 2\mu\biggr)\\
		&\qquad\qquad\leq\frac{\E\sup_{w\in K}\vert f_a(w)-\E f_a(w)\vert}{2\mu}=\frac{1}{2}.
	\end{align*}
	Using this inequality we get{\allowdisplaybreaks
	\begin{align*}
		&\phantom{\leq.}\frac{1}{2}\P\biggl(\sup\limits_{w\in K}\vert f_a(w)-\E f_a(w)\vert\geq 2\mu+t\biggr)\\
		&\leq \biggl(1-\P\biggl(\sup\limits_{w\in K}\vert f_a(w)-\E f_a(w)\vert\geq 2\mu\biggr)\biggr)\\
			&\qquad\cdot\P\biggl(\sup\limits_{w\in K}\vert f_a(w)-\E f_a(w)\vert\geq 2\mu+t\biggr)\\
		&=\P\bigg(\forall w\in K:\vert f_a(w)-\E f_a(w)\vert<2\mu\bigg)\\
			&\qquad\cdot\P\bigg(\exists w\in K:\vert f_a(w)-\E f_a(w)\vert\geq 2\mu+t\bigg).
	\end{align*}
	Let $\A_i$ and $\eps_i$ be again defined by \eqref{eq:bernoulli}, \eqref{eq:Ai_and_Ai'} and let $\A'_i$ be independent copies of $\A_i$. 
We further get
\begin{align*}
&\phantom{\leq.}\frac{1}{2}\P\biggl(\sup\limits_{w\in K}\vert f_a(w)-\E f_a(w)\vert\geq 2\mu+t\biggr)\\
&\le\P\bigg(\forall w\in K:\biggl\vert
\frac{1}{m}\sum\limits_{i=1}^m\biggl(\A_i(w)-\E\A(w)\biggr)\biggr\vert<2\mu\bigg)\\
&\phantom{=.}\cdot\P\bigg(\exists w\in K:\biggl\vert
\frac{1}{m}\sum\limits_{i=1}^m\biggl(\A_i'(w)-\E\A'(w)\biggr)\biggr\vert\geq2\mu+t\bigg)\\
&\leq\P\bigg(\exists w\in K:\biggr\vert\frac{1}{m}
\sum\limits_{i=1}^m\bigg(\big(\A_i(w)-\E \A(w)\big)\\
&\qquad\qquad-\big(\A_i'(w)-\E\A'(w)\big)\bigg)\biggl\vert\geq t\bigg)\\
&=\P\bigg(\exists w\in K:\biggr\vert\frac{1}{m}\sum\limits_{i=1}^m\eps_i(\A_i(w)-\A_i'(w))\biggl\vert\geq t\bigg)\\
&\leq2\P\bigg(\exists w\in K:\biggr\vert\frac{1}{m}\sum\limits_{i=1}^m\eps_i\A_i(w)\biggl\vert\geq t/2\bigg),
\end{align*}}
which yields the claim.
\end{proof}

Combining the Lemmas \ref{lemma:bernoulli} and \ref{lemma:deviation_inequality} we deduce the following result.

\begin{lemma}\label{lemma:final_uniform_concentration}
Under the ``Standing assumptions'' it holds for $\mu$ and $\tilde \mu$ according to \eqref{eq:mu} and \eqref{eq:tildemu} and any $u>0$ 
	\begin{align}\label{eq:uniform_concentration}
		&\phantom{.\leq}\P\biggl(\sup\limits_{w\in K}\vert f_a(w)-\E f_a(w)\vert
			\geq 2\mu+2\tilde \mu
			+u\biggr)\notag\\
		&\qquad\leq8\biggl(\exp\biggl(\frac{-mu^2}{32}\biggr)+\exp\biggl(\frac{-mu^2}{32r^2R^2}\biggr)\biggr).
	\end{align}
\end{lemma}
\begin{proof}
Applying Lemma \ref{lemma:deviation_inequality} and Lemma \ref{lemma:bernoulli} we get
\begin{align*}
&\phantom{\leq.}\P\biggl(\sup\limits_{w\in K}\vert f_a(w)-\E f_a(w)\vert\geq 2\mu+2\tilde\mu+u\biggr)\\
&\leq 4\P\biggl(\sup\limits_{w\in K}
\biggl\vert\frac{1}{m}\sum\limits_{i=1}^m\eps_i[1-y_i\sp{x_i}{w}]_+\biggr\vert\geq \tilde\mu+u/2\biggr)\\
&\leq 8\P\biggl(\sup\limits_{w\in K}\biggl\vert\frac{1}{m}\sum\limits_{i=1}^m\eps_i
(1-y_i\sp{x_i}{w})\biggr\vert\geq \tilde\mu+u/2\biggr)\\
&\leq 8\P\biggl(\biggl\vert\frac{1}{m}\sum\limits_{i=1}^m \eps_i\biggr\vert\geq u/4\biggr)\\
&\qquad+8\P\biggl(\sup\limits_{w\in K}\biggl\vert\biggl\langle\frac{1}{m}\sum\limits_{i=1}^m x_i,w
\biggr\rangle\biggr\vert\geq \tilde\mu+u/4\biggr).
\end{align*}
	Finally, applying the second and third part of Lemma \ref{lemma:concentration_bernoulli_gaussian} this can be further estimated from above by
	\begin{align*}
&\leq 8\biggl(\exp\biggl(\frac{-mu^2}{32}\biggr)+\exp\biggl(\frac{-mu^2}{32r^2R^2}\biggr)\biggr),
\end{align*}
which finishes the proof.
\end{proof}

Using the two Lemmas \ref{lemma:estimate_mu} and \ref{lemma:final_uniform_concentration} we can now prove 
Theorem \ref{theo:concentration_estimate}.

\begin{proof}[{\bf Proof of Theorem \ref{theo:concentration_estimate}}]
	Lemma \ref{lemma:final_uniform_concentration} yields
	\begin{align*}
		\P&\biggl(\sup\limits_{w\in K}\vert f_a(w)-\E f_a(w)\vert\geq 2\mu+2\tilde\mu+u\biggr)\notag\\
		&\qquad \leq 8\biggl(\exp\biggl(\frac{-mu^2}{32}\biggr)+\exp\biggl(\frac{-mu^2}{32r^2R^2}\biggr)\biggr).
	\end{align*}
	Using Lemma \ref{lemma:estimate_mu} we further get
	\begin{align*}
		\mu\leq \frac{4\sqrt{8\pi}+8rR\sqrt{2\log(2d)}}{\sqrt{m}}.
	\end{align*}
	Invoking the duality $\|\cdot\|_1'=\|\cdot\|_\infty$ and the first part of Lemma \ref{lemma:concentration_bernoulli_gaussian}
	we can further estimate $\tilde\mu$ by
	\begin{align*}
		\tilde\mu&
		=R\E\left\|\frac{1}{m}\sum\limits_{i=1}^m x_i\right\|_\infty
		\leq\frac{rR\sqrt{2\log(2d)}}{\sqrt{m}}.
	\end{align*}
	Hence, with probability at least
	\begin{align*}
		1-8\biggl(\exp\biggl(\frac{-mu^2}{32}\biggr)+\exp\biggl(\frac{-mu^2}{32r^2R^2}\biggr)\biggr)
	\end{align*}
	we have
	\begin{align*}
		\sup\limits_{w\in K}\vert f_a(w)-\E f_a(w)\vert
		&\leq 2\mu+2\tilde\mu+u\\
		&\leq \frac{8\sqrt{8\pi}+18rR\sqrt{2\log(2d)}}{\sqrt{m}}+u
	\end{align*}
	as claimed.
\end{proof}

\subsection{Estimate of the expected value}

In this subsection we will estimate 
\begin{align*}
	\E(f_a(w)-f_a(a))=\E[1-y\sp{x}{w}]_+-\E[1-y\sp{x}{a}]_+
\end{align*}
for some $w\in\R^d\backslash\{0\}$ with $\|w\|_1\leq R$.
We will first calculate both expected values separately and later estimate their difference.
We will make use of the following statements from probability theory.

\begin{lemma}\label{lemma:covariance}
	Let $a,x\in\R^d$ be according to \eqref{eq:assumptions-a}, \eqref{eq:assumptions-x}
	and let $w\in\R^d\backslash\{0\}$. Then it holds
	\begin{enumerate}
		\item $\sp{x}{a},~\sp{x}{\frac{w}{\|w\|_2}}\sim\mathcal{N}(0,r^2)$,
		\item $\cov(\sp{x}{a},\sp{x}{w})=r^2\sp{a}{w}$.
	\end{enumerate}
\end{lemma}
\begin{proof}
	The first statement is well known in probability theory as the 2-stability of normal distribution. For the second statement we get
	\begin{align*}
		\cov(\sp{x}{a},\sp{x}{w})
		&=\E(\sp{x}{a}\sp{x}{w})
		=\sum\limits_{i,j=1}^d a_iw_j\E(x_ix_j)\\
		&=r^2\sum\limits_{i=1}^da_iw_i=r^2\sp{a}{w}
	\end{align*}
	as claimed.
\end{proof}
It is very well known, cf. \cite[Corollary 5.2]{HS}, that projections of a Gaussian random vector onto two orthogonal directions are mutually independent.
\begin{lemma}\label{lemma:independent}
Let $x\sim {\mathcal N}(0,\id)$ and let $a,b\in\R^d$ with $\sp{a}{b}=0.$
Then $\sp{x}{a}$ and $\sp{x}{b}$ are independent random variables.
\end{lemma}

Applying these two lemmas to our case we end up with the following lemma.

\begin{lemma}\label{lemma:covariance_decomposition}
For $a\in\R^d$ according to \eqref{eq:assumptions-a}, $x\sim {\mathcal N}(0,r^2 \id)$ and $w\in\R^d$ we have
\begin{align*}
\sp{x}{w}=c\sp{x}{a}+c'Z
\end{align*}
for some $Z\sim\mathcal{N}(0,r^2)$ independent of $\sp{x}{a}$ and
\begin{align}\label{eq:constant_c}
c=\sp{a}{w},\quad c'=\sqrt{\|w\|_2^2-c^2}.
\end{align}
\end{lemma}
\begin{remark}
	Note that $c'$ is well defined, since $c^2\leq \|w\|_2^2\|a\|_2^2=\|w\|_2^2$.
\end{remark}
\begin{proof}
If $c'=0$, the statement holds trivially.
If $c'\not=0$, we set
\begin{align*}
Z=\frac{1}{c'}(\sp{x}{w}-c\sp{x}{a})=\frac{1}{c'}\sum\limits_{i=1}^d x_i\left(w_i-ca_i\right).
\end{align*}
Hence, $Z$ is indeed normally distributed with $\E(Z)=0$ and $\var(Z)=r^2$. It remains to show that $Z$ and $\sp{x}{a}$
are independent. 
We observe that
\begin{align*}
\sp{a}{w-ca}=\sp{a}{w}-\sp{a}{w}\|a\|_2=0
\end{align*}
and, finally, Lemma \ref{lemma:independent} yields the claim.
\end{proof}

\begin{lemma}\label{lemma:Psi_aa}
Let $a\in\R^d$ and $f_a\colon\R^d\to\R$ be according to \eqref{eq:assumptions-a}, \eqref{eq:f_a}. Then it holds
\begin{enumerate}
\item $\E f_a(a)
=\frac{1}{\sqrt{2\pi}}\int_\R\big[1-r\vert t\vert\big]_+ e^{\frac{-t^2}{2}}\,dt,
$
\item $\E f_a(w)=\frac{1}{2\pi}\int_{\R^2}\big[1-cr\vert t_1\vert-c'rt_2\big]_+e^{\frac{-t_1^2-t_2^2}{2}}\,dt_1\,dt_2$,
where $c$ and $c'$ are defined by \eqref{eq:constant_c}.
\end{enumerate}
\end{lemma}
\begin{proof}
\begin{enumerate}
\item Let $\omega\sim{\mathcal N}(0,1)$
and use the first part of Lemma \ref{lemma:covariance} to obtain
\begin{align*}
\E f_a(a)&=\E[1-\vert \sp{x}{a}\vert]_+=\E[1-r\vert\omega\vert]_+\\
&=\frac{1}{\sqrt{2\pi}}\int_\R\big[1-r\vert t\vert\big]_+ e^{\frac{-t^2}{2}}\,dt.
\end{align*}
\item Using the notation of Lemma \ref{lemma:covariance_decomposition}
we get
\begin{align*}
\E f_a(w)
&=\E[1-\sgn(\sp{x}{a})\sp{x}{w}]_+\\
&=\E[1-\sgn(\sp{x}{a})(c\sp{x}{a}+c'Z)]_+\\
&=\E[1-c\vert\sp{x}{a}\vert-c'\sgn(\sp{x}{a})Z]_+\\
&=\E[1-c\vert\sp{x}{a}\vert-c'Z]_+\\
&=\frac{1}{2\pi}\int_{\R^2}[1-cr\vert t_1\vert-c'rt_2]_+e^{\frac{-t_1^2-t_2^2}{2}}\,dt_1\,dt_2.
\end{align*}
\end{enumerate}
\end{proof}

Using this result we now can prove Theorem \ref{theo:Psi-Psi}.

\begin{proof}[{\bf Proof of Theorem \ref{theo:Psi-Psi}}]
Using Lemma \ref{lemma:Psi_aa} we first observe
\begin{align}\label{eq:estimate_Ef_a(a)}
-\pi\E f_a(a)
&=-\frac{\sqrt{\pi}}{\sqrt{2}}\int_{\R}[1-r|t|]_+e^{\frac{-t^2}{2}}\,dt\\
&=-\sqrt{2\pi}\int_0^{\frac{1}{r}}\big(1-rt\big)e^{\frac{-t^2}{2}}\,dt\notag\\
&\ge-\sqrt{2\pi}\int_0^{\frac{1}{r}}e^{\frac{-t^2}{2}}\,dt
\notag
\geq-\frac{\sqrt{2\pi}}{r}.
\end{align}
	To estimate the expected value of $f_a(w)$ we now distinguish the two cases $c\leq 0$ and $c>0$.\\
\emph{1. case: $c\leq0$:} In that case we get 
	\begin{align*}
		\pi\E f_a(w)
		&=\int_\R\int_0^\infty\big[1-crt_1-c'rt_2\big]_+e^{\frac{-t_1^2-t_2^2}{2}}\,dt_1\,dt_2.
	\end{align*}
	Since $-crt_1\geq0$ for $0\leq t_1<\infty$ we can further estimate
	\begin{align*}
		&\pi\E f_a(w)\geq\int_\R\int_0^\infty\big[1-c'rt_2\big]_+e^{\frac{-t_1^2-t_2^2}{2}}\,dt_1\,dt_2\\
		&\geq\int_{-\infty}^0\int_0^\infty(1-c'rt_2)e^{\frac{-t_1^2-t_2^2}{2}}\,dt_1\,dt_2\\
		&=\int_{-\infty}^0\int_0^\infty e^{\frac{-t_1^2-t_2^2}{2}}\,dt_1\,dt_2
			+c'r\int_0^\infty\int_0^\infty t_2e^{\frac{-t_1^2-t_2^2}{2}}\,dt_1\,dt_2\\
		&=\frac{\pi}{2}+c'r\frac{\sqrt{\pi}}{\sqrt{2}}.
	\end{align*}
	As claimed, putting both terms together, we arrive at
	\begin{align*}
		\pi\E(f_a(w)-f_a(a))
		\geq\frac{\pi}{2}+c'r\frac{\sqrt{\pi}}{\sqrt{2}}-\frac{\sqrt{2\pi}}{r}.
	\end{align*}
\emph{2. case: $c>0$:} First let us observe that 
$1-crt_1-c'rt_2\ge 0$ on $[0,1/cr]\times(-\infty,0]\subset\R^2$.
Hence, we get
\begin{align*}
	\pi\E f_a(w)
	&=\int_{\R^2}[1-crt_1-c'rt_2]_+e^{\frac{-t_1^2-t_2^2}{2}}\,dt_2\,dt_1\\
	&\geq\int_0^{\frac{1}{cr}}\int_{-\infty}^0(1-crt_1-c'rt_2)e^{\frac{-t_1^2-t_2^2}{2}}\,dt_2\,dt_1\\
	&=\frac{\sqrt{\pi}}{\sqrt{2}}\int_0^{\frac{1}{cr}}(1-crt)e^{\frac{-t^2}{2}}\,dt
	+c'r\int_0^{\frac{1}{cr}}e^{\frac{-t^2}{2}}\,dt\\
	&\geq \frac{\sqrt{\pi}}{\sqrt{2}}\int_0^{\frac{1}{cr}}(1-crt)e^{\frac{-t^2}{2}}\,dt
	+\frac{c'}{c}\exp\left(\frac{-1}{2c^2r^2}\right).
\end{align*}
	Combining this estimate with \eqref{eq:estimate_Ef_a(a)} we arrive at
	\begin{align*}
		\pi\E(f_a(w)-f_a(a))
		&\geq\frac{\sqrt{\pi}}{\sqrt{2}}\int_0^{\frac{1}{cr}}(1-crt)e^{\frac{-t^2}{2}}\,dt\\
			&\qquad+\frac{c'}{c}\exp\left(\frac{-1}{2c^2r^2}\right)-\frac{\sqrt{2\pi}}{r}.
	\end{align*}
\end{proof}

\section{$\l_1$-SVM with additional $\l_2$-constraint}\label{sec:4}

A detailed inspection of the analysis done so far shows that it would be convenient if the convex body $K$ would not
include vectors with large $\ell_2$-norm. For example, in \eqref{eq:mot:12} we needed to calculate $\sup_{w\in K}\|w\|_2^2=R^2$,
although the measure of the set of vectors in $K$ with $\ell_2$-norm close to $R$ is extremely small.

Therefore, we will modify the $\l_1$-SVM \eqref{eq:l1-svm} by adding an additional $\l_2$-constraint, that is instead of \eqref{eq:l1-svm} we consider the optimization problem
\begin{align}\label{eq:l2-l1-svm}
	\min\limits_{w\in\R^d}\sum\limits_{i=1}^m[1-y_i\sp{x_i}{w}]_+\ \text{s. t.}\ \|w\|_1\leq R~\text{and}~\|w\|_2\leq 1.
\end{align}
The combination of $\ell_1$ and $\ell_2$ constraints is by no means new - for example, it plays a crucial role
in the theory of elastic nets \cite{ZH}.
Furthermore, let us remark that the set
\begin{align}\label{eq:tildeK}
	\tilde K=\{w\in\R^d\mid \|w\|_1\leq R~\text{and}~\|w\|_2\leq 1\}
\end{align}
appears also in \cite{PV2}. We get $\tilde K\subset K$ with $K$ according to \eqref{eq:K}. Hence, Theorem \ref{theo:concentration_estimate}
and \eqref{eq:main_idea} still remain true if we replace $K$ by $\tilde K$ and we obtain
\begin{align}\label{eq:l2-concentration}
			\sup\limits_{w\in\tilde K}\vert f_a(w)-\E f_a(w)\vert
			\leq \frac{8\sqrt{8\pi}+18rR\sqrt{2\log(2d)}}{\sqrt{m}}+u
\end{align}
with high probability and
\begin{align}\label{eq:l2-main-idea}
		\E(f_a(\hat a)-f_a(a))
		\leq 2\sup\limits_{w\in\tilde K}\vert f_a(\hat a)-f_a(a)\vert,
\end{align}
where $\hat a$ is now the minimizer of \eqref{eq:l2-l1-svm}.

It remains to estimate the expected value $\E(f_a(w)-f_a(a))$ in order to obtain an analogue of
Theorem \ref{theo:main_theorem} for \eqref{eq:l2-l1-svm}, which reads as follows. 

\begin{theo}\label{theo:l2_main_theorem}
	Let $d\ge 2$, $0<\eps<1/2$, $r>2\sqrt{2\pi}(1-2\eps)^{-1}$, $a\in\R^d$ according to \eqref{eq:assumptions-a}, $m\geq C\eps^{-2}r^2R^2\log(d)$
	for some constant $C$, $x_1,\ldots,x_m\in\R^d$ according to \eqref{eq:assumptions-x} and $\hat a\in\R^d$
	a minimizer of \eqref{eq:l2-l1-svm}. Then it holds
	\begin{align}\label{eq:l2-l1-est}
	\|a-\hat a\|_2^2\leq \frac{C'\eps}{r(1-\exp\left(\frac{-1}{2r^2}\right))}
	\end{align}
	with probability at least
	\begin{align*}
		1-\gamma\exp\left(-C''\log(d)\right)
	\end{align*}
	for some positive constants $\gamma,C',C''$.
\end{theo}
\begin{remark}
\begin{enumerate}
\item As for Theorem \ref{theo:main_theorem} we can write down the expressions explicitly, 
	i.e. without the constants $\gamma, C,C'$ and $C''$. That is, taking
	$m\geq 4\eps^{-2}\left(8\sqrt{8\pi}+(18+t)rR\sqrt{2\log(2d)}\right)^2$ for some $t>0$, we get
	\begin{align*}
		\|a-\hat a\|_2^2\leq \frac{\sqrt{\pi/2}\ \eps}{r\left(1-\exp\left(\frac{-1}{2r^2}\right)\right)}.
	\end{align*}
	with probability at least
	\begin{align*}
		1-8\biggl(\exp\biggl(\frac{-t^2r^2R^2\log(2d)}{16}\biggr)+\exp\biggl(\frac{-t^2\log(2d)}{16}\biggr)\biggr).
	\end{align*}
\item The main advantage of Theorem \ref{theo:l2_main_theorem} compared to Theorem \ref{theo:main_theorem} is that the parameter $r$ does not need to grow
to infinity. Actually, \eqref{eq:l2-l1-est} is clearly not optimal for large $r$. Indeed, if (say) $\eps<0.2$, we can take $r=10$, and obtain
$$
\|a-\hat a\|_2^2\le \tilde C'\eps
$$
for $m\ge\tilde C\eps^{-2}R^2\log(d)$ with high probability.
\end{enumerate}
\end{remark}
\begin{proof}
	As in the proof of Theorem \ref{theo:main_theorem} we first obtain 
	$c'=\sqrt{\|\hat a\|_2^2-\sp{a}{\hat a}^2}>0$ and $c=\sp{a}{\hat a}>0$. Using Lemma \ref{lemma:Psi_aa}
	we get
	\begin{align*}
		\pi&\E(f_a(w)-f_a(a))\\
		&\geq\int_0^{\frac{1}{r}}\int_\R\big((1-crt_1-c'rt_2)-(1-rt_1)\big)e^{\frac{-t_1^2-t_2^2}{2}}\,dt_2\,dt_1\\
		&=r(1-c)\sqrt{2\pi}\int_0^{\frac{1}{r}}te^{\frac{-t^2}{2}}\,dt
	\end{align*}
	with
	\begin{align*}
		1-c=1-\sp{a}{\hat a}\geq \frac{1}{2}(\|a\|_2^2+\|\hat a\|_2^2)-\sp{a}{\hat a}=\frac{1}{2}\|a-\hat a\|_2^2.
	\end{align*}
	The claim now follows from \eqref{eq:l2-main-idea} and \eqref{eq:l2-concentration}.
\end{proof}

\section{Numerical experiments}\label{sec:5}

We performed several numerical tests to exhibit different aspects of the algorithms discussed above.
In the first two parts of this section we fixed $d=1000$ and set $\tilde a\in\R^d$ with 5 nonzero entries $\tilde a_{10}=1$, $\tilde a_{140}=-1$, $\tilde a_{234}=0.5$, 
$\tilde a_{360}=-0.5$, $\tilde a_{780}=0.3$,
Afterwards we normalized $\tilde a$ and set $a=\tilde a/\|\tilde a\|_2$ and $R=\|a\|_1$. 

\subsection{Dependency on $r$}

We run the $\ell_1$-SVM \eqref{eq:l1-svm} with $m=200$ and $m=400$ for different values of $r$ between zero and 1.5. The same was done for the $\ell_1$-SVM
with the additional $\ell_2$-constraint \eqref{eq:l2-l1-svm}, which is called $\ell_{1,2}$-SVM in the legend of the figure.
The average error of $n=20$ trials between $a$ and $\hat a/\|\hat a\|_2$ is plotted against $r$.
We observe that especially for small $r$'s the $\ell_1$-SVM with $\ell_2$-constraint performs much better than classical $\ell_1$-SVM.
\begin{figure}[ht]
	\centering
	\includegraphics[width=8cm]{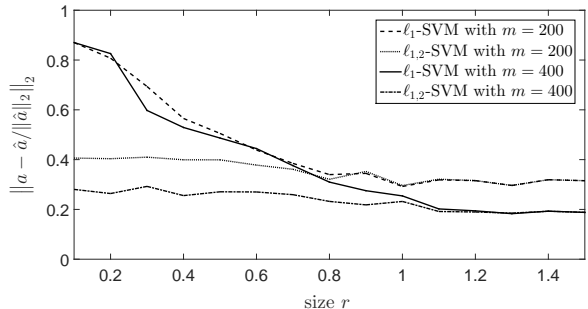}
	\caption{Dependency on $r$}
	\label{fig:dep_r}
\end{figure}

\subsection{Dependency on $m$ and comparison with 1-Bit CS}

In the second experiment, we run $\ell_1$-SVM with and without the extra $\ell_2$-constraint for two different values of $r$,
namely for $r=0.75$ and for $r$ depending on $m$ as $r=\sqrt{m}/30$. We plotted the average error of $n=40$ trials for each value.
The last method used is 1-bit Compressed Sensing \cite{PV2}, which is given as the maximizer of
\begin{align}\label{eq:1-Bit_CS}
	\max\limits_{w\in\R^d}\sum\limits_{i=1}^m y_i\sp{x_i}{w}\quad\text{subject to}\quad
	\|w\|_2\leq 1,~\|w\|_1\leq R.
\end{align}
Note that maximizer of \eqref{eq:1-Bit_CS} is independent of $r$, since it is linear in $x_i$.
\begin{figure}[ht]
	\centering
	\includegraphics[width=8cm]{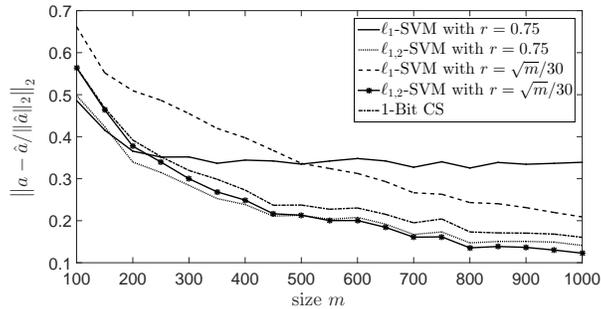}
	\caption{Comparison of $\l_1$-SVM with 1-Bit CS.}
	\label{fig:compare}
\end{figure}
First, one observes that the error of $\ell_1$-SVM does not converge to zero if the value of $r=0.75$ is fixed.
This is in a good agreement with Theorem \ref{theo:main_theorem} and the error estimate \eqref{eq:main_theorem}.
This drawback disappears when $r=\sqrt{m}/30$ grows with $m$, but $\ell_1$-SVM still performs quite badly.
The two versions of $\ell_{1,2}$-SVM perform essentially better than $\ell_1$-SVM, and slightly better than 1-bit Compressed Sensing.

\subsection{Dependency on $d$}
In figure \ref{fig:dep_d} we investigated the dependency of the error of $\l_1$-SVM on the dimension $d$.
We fixed the sparsity level $s=5$ and for each $d$ between $100$ and $3000$ we draw an $s$-sparse signal $a$ and measurement vectors $x_i$ at random.
Afterwards we run the $\l_1$-SVM with the three different values $m=m_i\log(d)$ with $m_1=10$, $m_2=20$ and $m_3=40$.
We plotted the average errors between $a$ and $\hat a/\|\hat a\|_2$ for $n=60$ trials.

	\begin{figure}[ht]
		\centering
		\includegraphics[width=8cm]{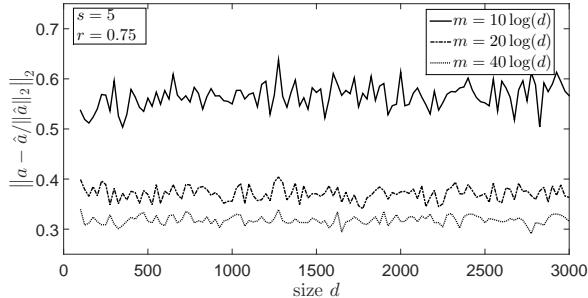}
		\caption{Dependency on $d$.}
		\label{fig:dep_d}
	\end{figure}
We indeed see that to achieve the same error, the number of measurements only needs to grow logarithmically in $d$, explaining
once again the success of $\ell_1$-SVM for high-dimensional classification problems.

\section{Discussion}

In this paper we have analyzed the performance of $\ell_1$-SVM \eqref{eq:l1-svm} in recovering sparse classifiers.
Theorem \ref{theo:main_theorem} shows, that a good approximation of such a sparse classifier can be achieved with small number
of learning points $m$ if the data is well spread. The geometric properties of well distributed learning points
are modelled by independent Gaussian vectors with growing variance $r$ and it would be interesting to know,
how $\ell_1$-SVM performs on points chosen independently from other distributions. The number of learning points
needs to grow logarithmically with the underlying dimension $d$ and linearly with the sparsity of the classifier.
On the other hand, the optimality of the dependence of $m$ on $\varepsilon$ and $r$ remains open.
Another important question left open is the behavior of $\ell_1$-SVM in the presence of missclasifications,
i.e. when there is a (small) probability that the signs $y_i\in\{-1,+1\}$ do not coincide with $\sgn(\langle x_i,a\rangle)$.
Finally, we proposed a modification of $\ell_1$-SVM by incorporating an additional $\ell_2$-constraint.

\section*{Acknowledgment}
\addcontentsline{toc}{section}{Acknowledgment}

We would like to thank A. Hinrichs, M. Omelka, and R. Vershynin for valuable discussions.

\begin{IEEEbiographynophoto}{Anton Kolleck}
received his M.S. in Mathematics at Technical University Berlin, Germany in 2013, where he now continues as Ph.D. student.
His research concentrates on sparse recovery and compressed sensing and their applications in approximation theory.
\end{IEEEbiographynophoto}

\begin{IEEEbiographynophoto}{Jan Vyb\'\i ral}
received his M.S. in Mathematics at Charles Univeristy, Prague, Czech Republic in 2002.
He earned the Dr. rer. nat. degree in Mathematics at Friedrich-Schiller University, Jena, Germany in 2005.
He had postdoc positions in Jena, Austrian Academy of Sciences, Austria, and Technical University Berlin, Germany.
He is currently an Assistant Professor of Mathematics at Charles University. His core interests are in functional analysis
with applications to sparse recovery and compressed sensing.
\end{IEEEbiographynophoto}
\vfill
\end{document}